\documentclass{amsart}

\newtheorem{theorem}{Theorem}[section]
\newtheorem{lemma}[theorem]{Lemma}

\theoremstyle{definition}

\theoremstyle{remark}
\newtheorem{remark}[theorem]{Remark}

\numberwithin{equation}{section}

\begin{document}

\title{On the Construction of Finite Oscillator Dictionary}

\author{Rongquan Feng}
\address{LMAM, School of Mathematical Sciences, Peking
University, Beijing 100871, P.R. China}
\email{fengrq@math.pku.edu.cn}
\thanks{The first author was supported by NSF of China (No. 10990011). The fourth author is supported by
China Postdoctoral Science Foundation funded project.}

\author{Zhenhua Gu}
\address{School of Mathematical Sciences, Suzhou
University,Suzhou 215006, P.R. China} \email{guzhmath@gmail.com}

\author{Zilong Wang}
\address{LMAM, School of Mathematical Sciences, Peking
University, Beijing 100871, P.R. China} \email{wzlmath@gmail.com}

\author{Hongfeng Wu}
\address{Academy of Mathematics and Systems Science, Chinese Academy of Sciences, Beijing 100190, P.R. China}
\email{whfmath@gmail.com}

\author{Kai Zhou}
\address{Academy of Mathematics and Systems Science, Chinese Academy of Sciences, Beijing 100190, P.R. China}
\email{kzhou@amss.ac.cn}





\begin{abstract}
A finite oscillator dictionary which has important applications in
sequences designs and the compressive sensing was introduced by
Gurevich, Hadani and Sochen. In this paper, we first revisit closed
formulae of the finite split oscillator dictionary $\mathfrak{S}^s$
by a simple proof. Then we study the non-split tori of the group
$SL(2,\mathbb{F}_p)$. Finally, An explicit algorithm for computing
the finite non-split oscillator dictionary $\mathfrak{S}^{ns}$ is
described.
\end{abstract}

\maketitle



\section{Introduction}
Let $\mathbb{F}_p$ ($p>3$) be the finite field with $p$ elements,
$\mathcal{H}=\mathbb{C}(\mathbb{F}_p)$  be a Hilbert space
containing all the functions from $\mathbb{F}_p$ to $\mathbb{C}$
with Hermitian product
\begin{equation*}
    \langle f,g\rangle=\sum_{t\in\mathbb{F}_p}f(t)\overline{g(t)},
\end{equation*}
for $f, g \in \mathcal{H}$, and let $U(\mathcal{H})$ be the group of
unitary operators on $\mathcal{H}$. Define $L_{\tau}$, $M_{\omega}$,
and $F\in U(\mathcal{H})$ by $$L_\tau f(t)=f(t+\tau),$$ $$M_\omega
f(t)=e^{\frac{2\pi i}{p}\omega t}f(t),$$ and
$$\hat{f}=F(f)(j)=\frac{1}{\sqrt{p}}\sum_{t\in
\mathbb{F}_p}e^{\frac{2\pi i}{p} jt}f(t),$$ for $\tau$,
$\omega\in\mathbb{F}_p$ and $f \in \mathcal{H}$. These operators
$L_{\tau}$, $M_{\omega}$ and $F$ are called the time shift, the
phase shift and the Fourier transform respectively which are
important operators in signal processings.

Denote by $SL(2,\mathbb{F}_p)$ the special linear group of
$2\times2$ nonsingular matrices over $\mathbb{F}_p$ with determinant
one. By studying the Weil representation $\rho:
SL(2,\mathbb{F}_p)\rightarrow U(\mathcal {H}),$ a finite oscillator
dictionary $\mathfrak{S}$ was given in \cite{GHS08} which has the
following properties:
\par (i) Autocorrelation (ambiguity function). For every
$\varphi\in\mathfrak{S}$,
\begin{equation*}
    |\langle\varphi,M_\omega L_\tau\varphi\rangle|=\left\{
    \begin{array}{ll}
        1, &  \text{ if }  \tau=\omega=0;\\
        \leq \frac{2}{\sqrt{p}}, & \text{ otherwise. }\\
    \end{array}
    \right.
\end{equation*}
\par (ii) Cross correlation (cross ambiguity function). For every
$\phi, \varphi\in\mathfrak{S}$, $\phi \neq \varphi$,
\begin{equation*}
    |\langle\phi,M_\omega L_\tau\varphi\rangle|\leq\frac{4}{\sqrt{p}},\ \ \ \text{ for
    every } \tau, \omega\in\mathbb{F}_p.
\end{equation*}
\par (iii) Supremum. For every $\varphi\in\mathfrak{S}$,
\begin{equation*}
    \max\{|\varphi(t)|:t\in\mathbb{F}_p\}\leq\frac{2}{\sqrt{p}}.
\end{equation*}
\par (iv) Fourier invariance. For every
$\varphi\in\mathfrak{S}$, its Fourier transform $\hat{\varphi}$ is
(up to multiplication by a unitary scalar) also in $\mathfrak{S}$.

The properties above make the finite oscillator dictionary
$\mathfrak{S}$ ideal for some applications. Please refer to
\cite{GHS08, Wang-Gong} for the application in the discrete radar
and communication systems \cite{Golomb-Gong}, and refer to
\cite{GHS082, GH10} for the application in the emerging fields of
sparsity and the compressive sensing. Besides, the Weil
representation on $SL(2,\mathbb{F}_p)$ provides a new approach for
the diagonalization of the discrete Fourier transform \cite{GH09,
Wang-Gong10}, and a proof for quadratic reciprocity \cite{GHH}.
Therefore, vectors in $\mathfrak{S}$ should be given in closed form
or by an algorithm. In fact, $\mathfrak{S}$ can be divided into two
parts naturally, split case $\mathfrak{S}^s$ and non-split case
$\mathfrak{S}^{ns}$. The split case $\mathfrak{S}^s$ was given by an
algorithm in \cite{GHS08}, and then closed formulae was given in
\cite{Wang-Gong}. In this paper, by studying the tori of
$SL(2,\mathbb{F}_p)$, closed formulae of the split case
$\mathfrak{S}^s$ were revisited by a simple proof, and an explicit
algorithm for constructing the non-split case $\mathfrak{S}^{ns}$ is
described.

The rest of the paper is organized as follows. The oscillator system
constructed in \cite{GHS08} is described in Section 2. Closed
formulae of the split case of the finite oscillator dictionary
$\mathfrak{S}^s$ are revisited in Section 3. In Section 4, based on
the stucture of the non-split tori of the group
$SL(2,\mathbb{F}_p)$, an explicit algorithm for computing the
non-split case of the finite oscillator dictionary
$\mathfrak{S}^{ns}$ is described.

\section{Finite Oscillator Dictionary}

In this section, the finite oscillator dictionary proposed by
Gurevich, Hadani and Sochen in \cite{GHS08, GHS082} was introduced.
For more details about the representation theory and the Weil
representation, we refer the reader to
\cite{Shapiro,Terras,Weil,Weyl}.

Since $SL(2, \mathbb{F}_p)$ is generated by $g_a=\left(
\begin{array}{cc}
a & 0 \\
0 & a^{-1} \\
\end{array}
\right)$, $g_b=\left(
\begin{array}{cc}
 1 & 0 \\
 b & 1 \\
 \end{array}
 \right)$, and the Weyl element
$w=\left(
\begin{array}{cc}
 0 & 1 \\
  -1 & 0 \\
  \end{array}
 \right)$,
where $a\in \mathbb{F}_p^*$ and $b \in \mathbb{F}_p$, the Weil
representation $\rho$ can be determined by $\rho(g_a), \rho(g_b)$
and $\rho(w)$ which are given as follows.
\begin{equation}\label{eq-sc1}
\rho(g_a)(f)(t)=\sigma(a)f(a^{-1}t),\\
\end{equation}
\begin{equation}\label{eq-ch1}
\rho(g_b)(f)(t)=\chi(-2^{-1}bt^2)f(t),
\end{equation}
\begin{equation}\label{eq-fo1}
\rho(w)(f)(j)=\frac{1}{\sqrt{p}}\sum_{t\in
\mathbb{F}_p}\chi(tj)f(t),
\end{equation}
where $\chi$ is an additive character of $\mathbb{F}_p$ with
$\chi(a)=e^{\frac{2\pi i}{p} a}$, and $\sigma$ is the Legendre
character, i.e., $\sigma (a)=(\frac{a}{p})$. Here $\rho(w)=F$ is the
discrete Fourier transform. Denote by $S_a$ the operator
$\rho(g_a)$, and by $N_b$ the operator $\rho(g_b)$ for convenience.
For any $g =\left(
\begin{array}{cc}
a & b \\
c & d \\
\end{array}
\right) \in SL_2(\mathbb{F}_p)$, if $b\neq 0$,
$$g=
 \left(
\begin{array}{cc}
a & b \\
(ad-1)b^{-1} & d \\
\end{array}
\right)= \left(
\begin{array}{cc}
b & 0 \\
0 & b^{-1} \\
\end{array}
\right) \left(
\begin{array}{cc}
1 & 0 \\
bd & 1 \\
\end{array}
\right) \left(
\begin{array}{cc}
0 & 1 \\
-1 & 0 \\
\end{array}
\right) \left(
\begin{array}{cc}
1 & 0 \\
ab^{-1} & 1 \\
\end{array}
\right).$$ Then the Weil representation of $g$ is given by
\begin{equation}
\rho(g)=S_{b}\circ N_{bd}\circ F\circ N_{ab^{-1}}.
\end{equation}
If $b=0$, then
$$g=\left(
\begin{array}{cc}
 a & 0 \\
 c & a^{-1} \\
 \end{array}
 \right)=
 \left(
\begin{array}{cc}
 a & 0 \\
 0 & a^{-1} \\
 \end{array}
 \right)\left(
\begin{array}{cc}
 1 & 0 \\
 ac & 1 \\
 \end{array}
 \right).$$
Hence the Weil representation of $g$ can be described as
\begin{equation}
\rho(g)=S_a\circ N_{ac}.
\end{equation}

A maximal {\em algebraic torus} in $SL(2, \mathbb{F}_p)$ is a
maximal commutative subgroup which becomes diagonalizable over the
original field $\mathbb{F}_p$ or over the quadratic extension of
$\mathbb{F}_p$. One standard example of a maximal algebraic torus in
$SL_2(\mathbb{F}_p)$ is the standard diagonal torus
$$A=\left\{\left(
\begin{array}{cc}
 a & 0 \\
 0 & a^{-1} \\
 \end{array}
 \right): a\in \mathbb{F}_p^*
\right\}.$$

Up to conjugation, there are two classes of maximal algebraic tori
in $SL(2, \mathbb{F}_p)$. The first class, called {\em split tori},
consists of those tori which are diagonalizable over $\mathbb{F}_p$.
Every split torus $T$ is conjugated to the standard diagonal torus
$A$, i.e., there exists an element $g \in SL(2, \mathbb{F}_p)$ such
that $g\cdot T \cdot g^{-1}= A$. The second class, called {\em
non-split tori}, consists of those tori which are not diagonalizable
over $\mathbb{F}_p$, but become diagonalizable over the quadratic
extension $\mathbb{F}_{p^2}$ of $\mathbb{F}_p$. In fact, a split
torus is a cyclic subgroup of $SL(2, \mathbb{F}_p)$ with order
$p-1$, while a non-split torus is a cyclic subgroup of $SL(2,
\mathbb{F}_p)$ with order $p+1$.

All split (non-split) tori are conjugated to one another, so the
number of split (non-split) tori equals to the number of elements in
the coset space $SL(2, \mathbb{F}_p)/{N_A}$ ($SL(2,
\mathbb{F}_p)/{N_T}$), where $N_A$ ($N_T$) is the normalizer group
of $A$ (non-split torus $T$). Thus
\begin{equation}
|(SL_2(\mathbb{F}_p)/{N_A})|=\frac{1}{2}p(p+1) \ \ \ \ \ \
\mbox{and}\ \ \ \ \ \ \
|(SL_2(\mathbb{F}_p)/{N_T})|=\frac{1}{2}p(p-1).
\end{equation}

Since every maximal torus $T\in SL(2, \mathbb{F}_p)$ is a cyclic
group, we obtain a decomposition of $\rho_{|T}: T\rightarrow
U(\mathcal{H})$, the restriction of the Weil representation $\rho$
on $T$, corresponding to an orthogonal decomposition of
$\mathcal{H}$ as
\begin{equation}\label{eq-DE}
\rho_{|T}=\bigoplus_{\chi\in \Lambda_T}\chi\ \ \ \ \mbox{and} \ \ \
\mathcal{H}=\bigoplus_{\chi\in\Lambda_T}\mathcal{H}_\chi,
\end{equation}
where $\Lambda_T$ is the collection of all the one dimensional
subrepresentation (character) $\chi: T\rightarrow \mathbb{C}$ of
$T$.

The decomposition (\ref{eq-DE}) depends on the type of $T$. If $T$
is a split torus, $\chi$ is a character given by $\chi:
\mathbb{Z}_{p-1}\rightarrow \mathbb{C}$. We have
$\dim\mathcal{H}_{\chi} = 1$ if $\chi$ is not the Legendre character
$\sigma$, and $\dim\mathcal{H}_{\sigma} = 2$. If $T$ is a non-split
torus,  then $\chi$ is a character given by $\chi:
\mathbb{Z}_{p+1}\rightarrow \mathbb{C}$. We have
$\dim\mathcal{H}_{\chi} = 1$ for every non-quadratic character
$\chi$.

For a given torus $T$, choosing a vector $\varphi_\chi\in
\mathcal{H}_\chi$ of unit norm for each character $\chi \in
\Lambda_T$, we obtain a collection of orthonormal vectors
\begin{equation}
\mathcal{B}_T=\{\varphi_\chi: \chi\in \Lambda_T, \chi\neq \sigma \
\mbox{if} \ T\  \mbox{is split}\}.
\end{equation}
Considering the union of all these collections, we obtain the finite
oscillator dictionary
\begin{equation}
\mathfrak{S}=\{\varphi\in \mathcal{B}_T: T \text{ is a maximal torus
of }SL(2,\mathbb{F}_p)\}=\bigcup_{T}\mathcal{B}_T,
\end{equation}
where $T$ runs through all maximal tori of $SL(2,\mathbb{F}_p)$.

The finite oscillator dictionary $\mathfrak{S}$ is naturally
separated into two sub-dictionaries $\mathfrak{S}^{s}$ and
$\mathfrak{S}^{ns}$ corresponding to the split and non-split tori
respectively. That is, $\mathfrak{S}^{s}$ ($\mathfrak{S}^{ns}$)
consists of the union of $\mathcal{B}_T$ where $T$ runs through all
the split tori (non-split tori) in $SL(2, \mathbb{F}_p)$. Totally
there are $\frac{1}{2}p(p+1)$ \ ($\frac{1}{2}p(p-1)$) split
(non-split) tori which consisting of $p-2$ ($p$) orthonormal vectors
each. Therefore
\begin{equation}
|\mathfrak{S}^{s}|=\frac{1}{2}p(p+1)(p-2)\ \ \ \ \mbox{and} \ \ \
|\mathfrak{S}^{ns}|=\frac{1}{2}p^2(p-1).
\end{equation}

For a given maximal torus $T$, an efficient way to specify the
decomposition (\ref{eq-DE}) is by choosing a generator $g_T\in T$,
the character is determined by the eigenvalue of the linear operator
$\rho(g_T)$, and the character space is corresponding to the
eigenspace naturally. Thus we can diagonalize $\rho(g_T)$ and obtain
the basis $\mathcal{B}_T$. Let $N_T$ be the normalizer of the group
$T$ and $R_T$ be a system of coset representatives of $N_T$ in
$SL(2,\mathbb{F}_p)$. Since all the maximal split or non-split tori
are conjugated to one another, all the maximal tori which are of the
same type with $T$ (split or non-split) can be written as
$gTg^{-1}$, where $g\in R_T$. Since
\begin{equation}
    \mathcal{B}_{gTg^{-1}}=\left\{\rho(g)\varphi:~\varphi\in\mathcal{B}_T \right\},
\end{equation}
the oscillator dictionary $\mathfrak{S}^s$ or $\mathfrak{S}^{ns}$
can be represented by
\begin{equation} \label{OS}
    \bigcup_{g\in R_T}\mathcal{B}_{gTg^{-1}}=\{\rho(g)\varphi: g\in R_T,\varphi\in\mathcal{B}_T\}.
\end{equation}

\section{The Split Case: $\mathfrak{S}^s$}

In this section, we revisit the results in \cite{Wang-Gong} by a
simple proof. Based on (\ref{OS}),
\begin{equation}\label{split}
\mathfrak{S}^s=\{\rho(g)\varphi:g\in R_A,\varphi\in\mathcal{B}_A\},
\end{equation}
we need only to compute $\mathcal{B}_A$ and $R_A$.

Let $\alpha$ be a generator of $\mathbb{F}_p^*$, define the
multiplicative character $\psi_j$ as $\psi_j(\alpha^k)=e^{\frac{2\pi
i}{p-1}jk}$ for $k=0,1,\cdots, p-2$ and $\psi_j(0)=0$ for $j\neq 0$
and $\psi_0(0)=1$. Then $\nabla=\left\{\psi_0, s\psi_1, \cdots,
s\psi_{p-2}\right\}$, where $s=(p-1)^{-1/2}$, is an orthonormal
basis of the Hilbert space $\mathcal{H}$. For $j\neq 0$, by
(\ref{eq-sc1}), we have
$$S_\alpha(s\psi_j)(\alpha^k)=\sigma(\alpha)s\psi_j(\alpha^{-1}\alpha^k)=-s\psi_j(\alpha^{k-1})=-se^{-\frac{2\pi ij}{p-1}}\psi_j(\alpha^k).$$
Noting that $S_\alpha(s\psi_j)(0)=0$, we have
$$S_\alpha(s\psi_j)=-e^{-\frac{2\pi i j}{p-1}}\cdot s\psi_j.$$
Thus $s\psi_j$ is an eigenvector of $S_\alpha$ associated with the
eigenvalue $-e^{-\frac{2\pi i j}{p-1}}\neq -1$. Therefore,
\begin{equation}
    \mathcal{B}_A=\{s\psi_j: 1\leq j\leq p-2\}.
\end{equation}

Combining the result in \cite{Wang-Gong} that
$$R_A=\left\{\left(
\begin{array}{cc}
1 & b \\
c & 1+bc \\
\end{array}
\right): 0 \leq b \leq \frac{p-1}{2}, c\in \mathbb{F}_p\right\}$$
and (\ref{split}), vectors in $\mathfrak{S}^s$ can be described in
closed formulae \cite{Wang-Gong} easily as
$$\mathfrak{S}^s=\{\varphi_{x,y,z}\,:\, \
1 \leq x\leq p-2, 0\leq y\leq p-1, 0\leq z\leq (p-1)/2\},$$ where
$$\varphi_{x,y,0}(t)=\frac{1}{\sqrt{p-1}}\psi_x(t)\chi(yt^2),$$
and
$$\varphi_{x,y,z}(t)=\frac{\chi(yt^2)}{\sqrt{p(p-1)}}\sum_{j=1}^{p-1}\psi_x(j)\chi(-(2z)^{-1}(j-t)^2)
\ \mbox{for} \ z\neq 0.$$

\section{The Non-Split Case: $\mathfrak{S}^{ns}$}
In this section we give the details of the construction of
$\mathfrak{S}^{ns}$. Let us first consider the structure of a
non-split torus.
\begin{lemma}Let $D$ be a non-square element of
$\mathbb{F}_p$. Then
\begin{equation*}
    T_D=\left\{\left(
            \begin{array}{cc}
              x & y \\
              Dy & x \\
            \end{array}
          \right):
          x^2-Dy^2=1,x,y\in\mathbb{F}_p
    \right\}
\end{equation*}
is a maximal non-split torus.
\end{lemma}
\begin{proof} Let
\begin{equation*}
    G_D=\left\{\left(
            \begin{array}{cc}
              x & y \\
              Dy & x \\
            \end{array}
          \right):
          x^2-Dy^2\neq0,x,y\in\mathbb{F}_p
    \right\}.
\end{equation*}
It is easy to check that $G_D$ is isomorphic to
$\mathbb{F}_{p^2}^{\ast}$ by the isomorphism given by
\begin{equation}\label{isomorphism}
    \left(
            \begin{array}{cc}
              x & y \\
              Dy & x \\
            \end{array}
          \right)\mapsto x+\sqrt{D}y.
\end{equation}
Thus $T_D$ can be diagonalized over $\mathbb{F}_{p^2}$ as a subgroup
of $G_D$. Note that the eigenvalues of $\left(
            \begin{array}{cc}
              x & y \\
              Dy & x \\
            \end{array}
          \right)\in T_D$ lie in $\mathbb{F}_{p^2}\setminus \mathbb{F}_{p}$. Hence
          $T_D$ is a non-split torus.
\end{proof}

\begin{theorem}\label{lemno} $(1)$ Let $D$ be a non-square element of
$\mathbb{F}_p$ and $s+\sqrt{D}t$ be a primitive element of
$\mathbb{F}_{p^2}$, then $T_D$ is a cyclic group of order $p+1$ with
a generator
\begin{equation*}
    \left(
       \begin{array}{cc}
         \frac{s^2+Dt^2}{s^2-Dt^2} & \frac{-2st}{s^2-Dt^2} \\
         \frac{-2stD}{s^2-Dt^2} & \frac{s^2+Dt^2}{s^2-Dt^2} \\
       \end{array}
     \right).
\end{equation*}
\par $(2)$ The normalizer of $T_D$ in $SL(2,\mathbb{F}_p)$ is
\begin{equation*}
    N_D=\left\{\left(
        \begin{array}{cc}
          a & b \\
          bD & a \\
        \end{array}
      \right),
      \left(
        \begin{array}{cc}
          x & y \\
          -yD & -x \\
        \end{array}
      \right):
      a^2-b^2D=y^2D-x^2=1,a,b,x,y\in\mathbb{F}_p
    \right\}.
\end{equation*}
\end{theorem}

\begin{proof} (1) By the isomorphism (\ref{isomorphism}), $\left(
                      \begin{array}{cc}
                        s & t \\
                        Dt & s \\
                      \end{array}
                    \right)$
is a generator of the cyclic group $G_D$. Since
$\frac{|G_D|}{|T_D|}=p-1$, $\left(
                      \begin{array}{cc}
                        s & t \\
                        Dt & s \\
                      \end{array}
                    \right)^{p-1}$
 is a generator of the cyclic subgroup $T_D$.
From
$$(s+\sqrt{D}t)^{p-1}=\frac{(s+\sqrt{D}t)^{p}}{s+\sqrt{D}t}=
 \frac{s-\sqrt{D}t}{s+\sqrt{D}t}=\frac{s^2+Dt^2-2st\sqrt{D}}{s^2-Dt^2},$$
we know that $$
    \left(
                      \begin{array}{cc}
                        s & t \\
                        Dt & s \\
                      \end{array}
                    \right)^{p-1}=\left(
       \begin{array}{cc}
         \frac{s^2+Dt^2}{s^2-Dt^2} & \frac{-2st}{s^2-Dt^2} \\
         \frac{-2stD}{s^2-Dt^2} & \frac{s^2+Dt^2}{s^2-Dt^2} \\
       \end{array}
     \right)
$$
is a generator of $T_D$.
\par (2) Suppose $g=\left(
                       \begin{array}{cc}
                         a & b \\
                         c & d \\
                       \end{array}
                     \right)\in N_D.$ Then for every $\left(
            \begin{array}{cc}
              x & y \\
              Dy & x \\
            \end{array}
          \right)\in T_D$,
\begin{equation*}
    \left(
                       \begin{array}{cc}
                         a & b \\
                         c & d \\
                       \end{array}
                     \right)
    \left(
                       \begin{array}{cc}
                         x & y \\
                         Dy & x \\
                       \end{array}
                     \right)
    \left(
                       \begin{array}{cc}
                         a & b \\
                         c & d \\
                       \end{array}
                     \right)^{-1}
    =
    \left(
                       \begin{array}{cc}
                         x+bdyD-acy & a^2y-b^2yD \\
                         d^2yD-c^2y & x+acy-bdyD \\
                       \end{array}
                     \right)\in T_D,
\end{equation*}
which implies
\begin{equation}\label{ssss-eqncN1}
    \left\{
\begin{array}{rcl}
 ac-bdD &=& 0, \\
  c^2-d^2D &=& b^2D^2-a^2D.
\end{array}
\right.
\end{equation}
Combining (\ref{ssss-eqncN1}) with
\begin{equation}\label{ssss-eqncN2}
    ad-bc=1,
\end{equation}
we have $g=\left(
\begin{array}{cc}
a & b \\
bD & a \\
\end{array}
\right)$ with $a^2-b^2D=1$ or $g=\left(
\begin{array}{cc}
a & b \\
-bD & -a \\
\end{array}
\right)$ with $a^2-b^2D=-1$. Therefore,
   \begin{equation*}
    N_D=\left\{\left(
        \begin{array}{cc}
          a & b \\
          bD & a \\
        \end{array}
      \right),
      \left(
        \begin{array}{cc}
          x & y \\
          -yD & -x \\
        \end{array}
      \right):
      a^2-b^2D=1, y^2D-x^2=1,a,b,x,y\in\mathbb{F}_p
    \right\}.
\end{equation*}
\end{proof}

Denote by $\pm \sqrt{-1}$ the two roots of the equation $X^2+1=0$.
If $p\equiv 3$ (mod 4), $-1$ is a non-square element, then $\pm
\sqrt{-1} \in \mathbb{F}_{p^2}\setminus \mathbb{F}_{p}$. If $p\equiv
1$ (mod 4), $-1$ is a square element and $\pm \sqrt{-1} \in
\mathbb{F}_p$. Now we define $S$ as a subset of $\mathbb{F}_p$
satisfying (i) for every $x\in S$, $1\leq x\leq \frac{p-1}{2}$, and
(ii) $x \in S \Leftrightarrow \pm \sqrt{-1}x \notin S$. It is
obvious that $|S|=\frac{p-1}{4}$.

\begin{theorem}$(1)$ If $p\equiv 3\  ({\rm mod}\;4)$, let $$R_D=\left\{\left(
                \begin{array}{cc}
                  a & 0 \\
                  c & a^{-1} \\
                \end{array}
              \right): 1\leq a \leq \frac{p-1}{2},0\leq c\leq p-1
\right\},$$ then $R_D$ is a collection of coset representatives of
$N_D$ in $SL(2,\mathbb{F}_p)$.

\par $(2)$ If $p\equiv 1\  ({\rm mod}\;4)$, let $$R^{'}_D=\left\{\left(
                \begin{array}{cc}
                  a & 0 \\
                  c & a^{-1} \\
                \end{array}
              \right), \left(\begin{array}{cc}
                  a & 0 \\
                  c & a^{-1} \\
                \end{array}
              \right)
              \left(\begin{array}{cc}
                  0 & 1 \\
                  -1 & 0 \\
                \end{array}
              \right): a \in S,0\leq c\leq p-1
\right\},$$ then $R^{'}_D$ is a collection of coset representatives
of $N_D$ in $SL(2,\mathbb{F}_p)$.
\end{theorem}
\begin{proof} (1) For $p\equiv 3\  ({\rm mod}\;4)$, suppose that $g=\left(
                       \begin{array}{cc}
                         a & b \\
                         c & d \\
                       \end{array}
                     \right)\in SL(2,\mathbb{F}_p)$. Since $-1$ is a non-square element, one of the equations $\xi^2=\frac{b^2}{b^2D-a^2}$
                     and
$\xi^2=\frac{b^2}{a^2-b^2D}$ is solvable over $\mathbb{F}_p$ for
$b\neq 0$. If $y$ is a root of $\xi^2=\frac{b^2}{b^2D-a^2}$, let
$x=\frac{ay}{b}$, then
\begin{equation} \label{matrix}
         \left(
             \begin{array}{cc}
              a & b \\
              c & d \\
             \end{array}
         \right)
    \left(
    \begin{array}{cc}
    x & y \\
    -Dy & -x \\
    \end{array}
    \right)
    =
    \left(
    \begin{array}{cc}
     ax-byD & 0 \\
     cx-dyD & cy-dx \\
     \end{array}
     \right).
\end{equation}
If $y$ is a root of $\xi^2=\frac{b^2}{a^2-b^2D}$, let
$x=-\frac{ay}{b}$, then

\begin{equation*}
         \left(
             \begin{array}{cc}
              a & b \\
              c & d \\
             \end{array}
         \right)
    \left(
                       \begin{array}{cc}
                         x & y \\
                         Dy & x \\
                       \end{array}
                     \right)
    =
    \left(
                       \begin{array}{cc}
                         ax+byD & 0\\
                         cx+dyD & cy+dx \\
                       \end{array}
                     \right).
\end{equation*}
If $b=0$, then
\begin{equation*}
         \left(
             \begin{array}{cc}
              a & b \\
              c & d \\
             \end{array}
         \right)
    =
    \left(
                       \begin{array}{cc}
                         a & 0 \\
                         c & a^{-1} \\
                       \end{array}
                     \right).
\end{equation*}
Therefore, every coset of $N_D$ in $SL(2,\mathbb{F}_p)$ has a
representative of the form $\left(
                                     \begin{array}{cc}
                                       a & 0 \\
                                       c & a^{-1} \\
                                     \end{array}
                                   \right)$. If two such matrices
$\left(
                                     \begin{array}{cc}
                                       a & 0 \\
                                       c & a^{-1} \\
                                     \end{array}
                                   \right)$
and $\left(
                                     \begin{array}{cc}
                                       x & 0 \\
                                       y & x^{-1} \\
                                     \end{array}
                                   \right)$
are in the same coset, then

\begin{equation*}
\left(
    \begin{array}{cc}
         a & 0 \\
         c & a^{-1} \\
          \end{array}
\right)^{-1} \left(
    \begin{array}{cc}
    x & 0 \\
    y & x^{-1} \\
    \end{array}
\right) = \left(
  \begin{array}{cc}
    xa^{-1} & 0 \\
    ay-cx & ax^{-1} \\
  \end{array}
\right)\in N_D.
\end{equation*}
Thus we have

\begin{equation}\label{sss-eqqwn1}
\left\{\begin{array}{rcl}
  xa^{-1} &=& ax^{-1}, \\
  ay &=& cx,
\end{array}
\right.
\end{equation}
or
\begin{equation}\label{sss-eqqwn2}
\left\{\begin{array}{rcl}
  xa^{-1} &=& -ax^{-1}, \\
  ay &=& cx,
\end{array}
\right.
\end{equation}
Equation (\ref{sss-eqqwn2}) gives $(xa^{-1})^2=-1$, which is
impossible since $p\equiv 3\  ({\rm mod}\;4) $. From
(\ref{sss-eqqwn1}), we have
\begin{equation*}
\left\{\begin{array}{rcl}
  a^2 &=& x^2, \\
  ay &=& cx,
\end{array}
\right.
\end{equation*}
which implies
\begin{equation*}
\left\{\begin{array}{c}
  a=x,\\
  c=y,
\end{array}
\right. \text{~or~} \left\{\begin{array}{c}
  a=-x,\\
  c=-y.
\end{array}
\right.
\end{equation*}
Therefore $R_D$ is a collection of coset representatives of $N_D$ in
$SL(2,\mathbb{F}_p)$.

(2)For  $p\equiv 1\  ({\rm mod}\;4) $, suppose that $g=\left(
                       \begin{array}{cc}
                         a & b \\
                         c & d \\
                       \end{array}
                     \right)\in SL(2,\mathbb{F}_p)$.
If $\xi^2=\frac{b^2}{b^2D-a^2}$ is solvable over $\mathbb{F}_p$,
then by (\ref{matrix}), there exist a lower triangle matrix, such
that $\left(
                       \begin{array}{cc}
                         a & b \\
                         c & d \\
                       \end{array}
                     \right)$ and this lower triangle matrix are in the
same coset.

Otherwise $\xi^2=\frac{a^2}{a^2D-b^2D^2}$ is solvable over
$\mathbb{F}_p$. let $y$ be a root of $\xi^2=\frac{b^2}{b^2D-a^2}$
and  $x=\frac{byD}{a}$, then
\begin{equation*}
         \left(
             \begin{array}{cc}
              a & b \\
              c & d \\
             \end{array}
         \right)
    \left(
    \begin{array}{cc}
    x & y \\
    -Dy & -x \\
    \end{array}
    \right)
    =
    \left(
    \begin{array}{cc}
     0 & ay-bx \\
     cx-dyD & cy-dx \\
     \end{array}
     \right).
\end{equation*}

Therefore, every coset representative has the form $\left(
                                     \begin{array}{cc}
                                       a & 0 \\
                                       c & a^{-1} \\
                                     \end{array}
                                   \right)$
or $\left(
                                     \begin{array}{cc}
                                       0 & x^{-1}\\
                                       -x & y \\
                                     \end{array}
                                   \right)$.

If two such matrices $\left(
                                     \begin{array}{cc}
                                       a & 0 \\
                                       c & a^{-1} \\
                                     \end{array}
                                   \right)$
and $\left(
                                     \begin{array}{cc}
                                       x & 0 \\
                                       y & x^{-1} \\
                                     \end{array}
                                   \right)$
are in the same coset, then equations (\ref{sss-eqqwn1}) and
(\ref{sss-eqqwn2}) imply
\begin{equation*}
\left\{\begin{array}{c}
  a=x,\\
  c=y,
\end{array}
\right. \text{~or~} \left\{\begin{array}{c}
  a=-x,\\
  c=-y,
\end{array}
\right. \text{~or~} \left\{\begin{array}{c}
  a=\sqrt{-1}x,\\
  c=\sqrt{-1}y,
\end{array}
\right. \text{~or~} \left\{\begin{array}{c}
  a=-\sqrt{-1}x,\\
  c=-\sqrt{-1}y,
\end{array}
\right.
\end{equation*} where $\sqrt{-1}$ is the smaller root of the equation $X^2+1=0$. Therefore $$\left\{\left(
                \begin{array}{cc}
                  a & 0 \\
                  c & a^{-1} \\
                \end{array}
              \right): a \in S,0\leq c\leq p-1
\right\}$$ includes $\frac{p(p-1)}{4}$ different coset
representatives of $N_D$ in $SL(2,\mathbb{F}_p)$.

Similarly,
$$\left\{\left(
                \begin{array}{cc}
                  0 & x^{-1} \\
                  -x & y \\
                \end{array}
              \right): x \in S,0\leq y\leq p-1
\right\}$$ includes $\frac{p(p-1)}{4}$ coset different
representative elements of $N_D$ in $SL(2,\mathbb{F}_p)$.
\par
If two such matrices $\left(
                                     \begin{array}{cc}
                                       a & 0 \\
                                       c & a^{-1} \\
                                     \end{array}
                                   \right)$
and $\left(
                                     \begin{array}{cc}
                                       0 & x^{-1}\\
                                       -x & y \\
                                     \end{array}
                                   \right)$
are in the same coset, then
\begin{equation*}
\left(
    \begin{array}{cc}
         a & 0 \\
         c & a^{-1} \\
          \end{array}
\right)^{-1} \left(
    \begin{array}{cc}
    0 & x^{-1} \\
    -x & y \\
    \end{array}
\right) = \left(
  \begin{array}{cc}
    0 & a^{-1}x^{-1} \\
    -ax & -cx^{-1}+ay \\
  \end{array}
\right)\in N_D.
\end{equation*}
Thus we have $a^{-1}x^{-1}=-Dax$ or $a^{-1}x^{-1}=Dax$, which
implies $-D$ or $D$ is a square element of $\mathbb{F}_p$. Both of
them are impossible since $p\equiv 1\ ({\rm mod}\;4)$.

Therefore,
\begin{eqnarray*}
R^{'}_D&=&\left\{\left(
                \begin{array}{cc}
                  a & 0 \\
                  c & a^{-1} \\
                \end{array}
              \right), \left(
                \begin{array}{cc}
                  0 & x^{-1} \\
                  -x & y \\
                \end{array}
              \right): a,x \in S,0\leq c,y\leq p-1
\right\}\\
&=&\left\{\left(
                \begin{array}{cc}
                  a & 0 \\
                  c & a^{-1} \\
                \end{array}
              \right), \left(\begin{array}{cc}
                  a & 0 \\
                  c & a^{-1} \\
                \end{array}
              \right)
              \left(\begin{array}{cc}
                  0 & 1 \\
                  -1 & 0 \\
                \end{array}
              \right): a \in S,0\leq c\leq p-1
\right\}
\end{eqnarray*}
is a collection of coset representatives of $N_D$ in
$SL(2,\mathbb{F}_p)$.
\end{proof}

In the following, we give an algorithm for constructing the
non-split finite oscillator dictionary $\mathfrak{S}^{ns}$.

\par \noindent{\bf Algorithm for $\mathfrak{S}^{ns}$}
\begin{enumerate}
  \item For a given $p$, choose a non-square element $D$ of $\mathbb{F}_p$, and $s$ and $t$ such that $
s+t\sqrt{D}$ is a primitive element of $\mathbb{F}_{p^2}$.
  \item Let $g_D=\left(
       \begin{array}{cc}
         \frac{s^2+Dt^2}{s^2-Dt^2} & \frac{-2st}{s^2-Dt^2} \\
         \frac{-2stD}{s^2-Dt^2} & \frac{s^2+Dt^2}{s^2-Dt^2} \\
       \end{array}
     \right).$ Diagonalize $\rho(g_D)$ to obtain
     $\mathcal{B}_T$.
  \item If $p\equiv 3\  ({\rm mod}\;4)$, then
  $$\mathfrak{S}^{ns}=\left\{S_a\circ N_{ac}(\varphi):~\varphi\in\mathcal{B}_T, 1\leq a \leq \frac{p-1}{2},0\leq c\leq p-1 \right\}.$$

If $p\equiv 1\  ({\rm mod}\;4)$, then
  $$\mathfrak{S}^{ns}=\left\{S_a\circ N_{ac}(\varphi):~\varphi\in\mathcal{B}_T \cup F\circ \mathcal{B}_T, a \in S,0\leq c\leq p-1 \right\}.$$
\end{enumerate}

\begin{remark}
The above results can be easily generalized from $\mathbb{F}_p$ to
$\mathbb{F}_{p^n}$.
\end{remark}

\bibliographystyle{amsplain}

\begin{thebibliography}{10}

\bibitem{Bump} D. Bump, {\em Automorphic Forms and Representations}, Cambridge University
Press, Cambridge, 1998.

\bibitem{Golomb-Gong} S.W. Golomb and G. Gong, {\em Signal Design with Good Correlation: for
Wireless Communications, Cryptography and Radar Applications},
Cambridge University Press, Cambridge, 2005.

\bibitem{GHS08} S. Gurevich, R. Hadani, and N. Sochen, The finite harmonic oscillator and its applications to sequences,
Communication and radar, {\em IEEE Trans. Inform. Theory}, Vol 54,
No.9, Sep. 2008, pp. 4239-4253.

\bibitem{GHS082} S. Gurevich, R. Hadani, and  N. Sochen, On some deterministic
dictionaries supporting sparsity, {\em Journal of Fourier Analysis
and Applications}, Vol. 14, No. 5-6, Dec. 2008, pp. 859-876.

\bibitem{GH09} S. Gurevich, R. Hadani, On the diagonalization of the discrete
Fourier transform, {\em Applied and Computational Harmonic
Analysis}, Vol. 27, Issue 1, July 2009, pp. 87-99.

\bibitem{GHH}S. Gurevich, R. Hadani, R. Howe, Quadratic reciprocity and sign of
Gauss sum via the finite Weil representation, accpeted by {\em
IMRN}. http://arxiv.org/abs/0808.2447.

\bibitem{GH10} S. Gurevich, R. Hadani, The statistical restricted isometry
property and the Wigner semicircle distribution of incoherent
dictionaries. Submitted to the {\em Annals of Applied Probability}
2009.

\bibitem{howe} R. Howe, Nice error bases, mutually unbiased bases, induced
representations, the Heisenberg group and finite geometries, {\em
Indag. Math.(N.S.)}, vol.16, no.3-4, 2005, pp. 553-583.

\bibitem{howard} S. D. Howard, A.R. Calderbank, and W.\ Moran, The finite Heisenberg-
Weyl groups in radar and communications, {\em EURASIP J. Appl.
Signal Process}., 2006:85685, 2006.

\bibitem{Shapiro} I. I. Piatetski-Shapiro, Complex representations of $GL(2,K)$ for
finite fields $K$, {Contemporary Math.}, 16, Amer. Math. Soc.,
providence, 1983.

\bibitem{Terras} A. Terras, {\em Fourier Analysis on Finite Groups and Its
Applications}, London Mathematical Society Student Texts 43,
Cambridge University Press, Cambridge, 1999.

\bibitem{Wang-Gong}Z. Wang, G. Gong, New sequences design from Weil representation
with low two-dimensional correlation in both time and phase shifts,
http://arxiv.org/abs/0812.4487, Technical Report 2009-1, University
of Waterloo, 2009.

\bibitem{Wang-Gong10}Z. Wang, G. Gong, A note on the diagonalization of the discrete
Fourier transform, {\em Applied and Computational Harmonic
Analysis}, Vol. 28, Issue 1, Jan. 2010, pp.114-120.


\bibitem{Weil} A. Weil, Sur certains groupes  d'op\'{e}arateurs unitaires, {\em Acta
Math.}, vol 111, 1964, pp. 143-211.

\bibitem{Weyl} H.\ Weyl, {\em The Classical Groups. Their Invariants and
Representations}, Princeton, N.J.: Princeton Univ. Press, 1939.

\end{thebibliography}

\end{document}